\documentclass[11pt]{llncs}
\usepackage{makeidx}
\usepackage{graphicx}
\usepackage{xcolor}

\usepackage{ifthen}
\usepackage{microtype}%if unwanted, comment out or use option "draft"
\usepackage{boxedminipage}
\usepackage{amsmath}
\usepackage{amsfonts}
\usepackage{amssymb}
\usepackage{authblk}
\usepackage{breakcites}
\usepackage{chemscheme}
\usepackage{hyperref}

\usepackage{float}
\newfloat{Algorithm}{hbtp}{lop}[section]

\newcommand{\lv}[1]{#1}
\newcommand{\sv}[1]{}

\renewcommand{\sp}{{\hspace*{1 mm}}}

\newcommand{\pr}{\mathbf{Pr}}
\newcommand{\eps}{{\varepsilon}}

\newcommand{\poly}{\mathrm{poly}}

\newcommand{\nl}{\bot}
\newcommand{\veps}{\varepsilon}
\renewcommand{\poly}{{\mbox {\small poly}}}

\newcommand{\ost}[1]{O^{\star}_{#1}}

\newcommand{\cost}{{\tt cost}}
\newcommand{\opt}{{\tt opt}}
\newcommand{\I}{{\cal I}}

\newcommand{\T}{{\cal T}}
\renewcommand{\L}{{\cal L}}
\newcommand{\ti}{{\bar i}}
\newcommand{\ostp}[1]{{O_{#1}^{'}}}
\newcommand{\ostn}[1]{{O_{#1}^{n}}}
\newcommand{\ostpn}[1]{{c(O_{#1}^{n})}}
\newcommand{\ostf}[1]{{O_{#1}^{f}}}
\newcommand{\core}{{\tt core}}
\newcommand{\const}{{\tt construct}}

\newcommand{\bO}{{\mathbb O}}
\newcommand{\bOst}{{\mathbb O^{\star}}}

\newcommand{\ctd}{{\mathbb C}}
\newcommand{\actd}{{A^{\ctd}}}

\spnewtheorem{fact}{Fact}{\bfseries}{\itshape}
\allowdisplaybreaks

\graphicspath{{./Figures/}}

\begin{document}

\title{Faster Algorithms for the Constrained $k$-means Problem}
%
%\titlerunning{}
% abbreviated title (for running head)
\author{Anup Bhattacharya \and Ragesh Jaiswal \and Amit Kumar}
%
%\authorrunning{Jaiswal et al.}
% abbreviated author list (for running head)
%\tocauthor{Ragesh Jaiswal}
%
\institute{Department of Computer Science and Engineering, \\
Indian Institute of Technology Delhi.\thanks{Email addresses: \email{\{anupb, rjaiswal, amitk\}@cse.iitd.ac.in}}
}
\sv{\thispagestyle{empty}}
\maketitle     % typeset the title of the contribution
\begin{abstract}
The classical center based clustering problems such as $k$-means/median/center assume that the optimal clusters satisfy the locality property that the points in the same cluster are close to each other.
A number of clustering problems arise in machine learning where the optimal clusters do not follow such a locality property.
For instance, consider the {\em $r$-gather clustering problem} where there is an additional constraint that each of the clusters should have at least $r$ points or the {\em capacitated clustering problem} where there is an upper bound on the cluster sizes.
Consider a variant of the $k$-means problem that may be regarded as a general version of such problems.
Here, the optimal clusters $O_1, ..., O_k$ are an {\em arbitrary} partition of the dataset and the goal is to output $k$-centers $c_1, ..., c_k$ such that the objective function $\sum_{i=1}^{k} \sum_{x \in O_{i}} ||x - c_{i}||^2$ is minimized.
It is not difficult to argue that any algorithm (without knowing the optimal clusters) that outputs a single set of $k$ centers, will not behave well as far as optimizing the above objective function is concerned.
However, this does not rule out the existence of algorithms that output a {\em list} of such $k$ centers such that at least one of these $k$ centers behaves well.
Given an error parameter $\veps > 0$, let $\ell$ denote the size of the smallest list of $k$-centers such that at least one of the $k$-centers gives a $(1+\veps)$ approximation w.r.t. the objective function above.
In this paper, we show an upper bound on $\ell$ by giving a randomized algorithm that outputs a list of $2^{\tilde{O}(k/\veps)}$ $k$-centers
\footnote{$\tilde{O}$ notation hides a $O(\log{\frac{k}{\eps}})$ factor.}.
We also give a closely matching lower bound of $2^{\tilde{\Omega}(k/\sqrt{\veps})}$.
Moreover, our algorithm runs in time $O \left(n d \cdot 2^{\tilde{O}(k/\veps)} \right)$.
This is a significant improvement over the previous result of Ding and Xu~\cite{dx15} who gave an algorithm with running time $O \left(n d \cdot (\log{n})^{k} \cdot 2^{poly(k/\veps)} \right)$ and output a list of size $O \left((\log{n})^k \cdot 2^{poly(k/\veps)} \right)$.
Our techniques generalize for the $k$-median problem and for many other settings where non-Euclidean distance measures are involved.

\end{abstract}
%\pagebreak
%%\newpage
%\setcounter{page}{1}
\section{Introduction}
Clustering problems intend to classify high dimensional data based on the proximity of points to each other.
There is an inherent assumption that the clusters satisfy  {\em locality} property -- points close to each other (in a geometric sense) should belong to the same category.
Often, we model such problems by the notion of a center based clustering problem.
We would like to identify a set of centers, one for each cluster, and then the clustering is obtained by assigning each point to the nearest center.
For example, the $k$-means problem is defined in the following manner: given a dataset $X = \{x_1, \ldots, x_n\} \subset \mathbb{R}^d$ and an integer $k$, output a set of $k$ centers $\{c_1, \ldots, c_k\} \subset \mathbb{R}^d$ such that the objective function $\sum_{x \in X} \min_{c \in \{c_1, \ldots, c_k\}} ||x - c||^2$ is minimized.
The $k$-median and the $k$-center problems are defined in a similar manner by defining a suitable objective function.

However, often such clustering problems entail several {\em side constraints}.
Such constraints limit the set of feasible clusterings.
For example, the $r$-gather $k$-means clustering problem is defined in the same manner as the $k$-means problem, but has the additional constraint that each cluster must have at least $r$ points in it.
In such settings, it is no longer true that the clustering is obtained from the set of centers by the Voronoi partition.
Ding and Xu~\cite{dx15} began a systematic study of such problems, and this is the
starting point of our work as well.
They defined the so-called {\em constrained $k$-means} problem.
An instance of such  a problem is specified by a set of points $X$, a parameter $k$, and a set $\ctd$, where each element of $\ctd$ is a partitioning of $X$ into $k$ disjoint subsets (or clusters).
Since the set $\ctd$ may be exponentially large, we will assume that it is specified in a
succinct manner by an efficient algorithm which decides membership in this set.
A solution needs to output an element $\bO = \{O_1, \ldots, O_k\}$ of $\ctd$, and a set of $k$ centers, $c_1, \ldots, c_k$, one for each cluster in $\bO$.
The goal is to minimize $\sum_{i=1}^k \sum_{x \in O_i} ||x-c_i||^2$. It is easy to check that the center $c_i$ must be the mean of the corresponding cluster $O_i$.
Note that the $k$-means problem is a special case of this problem where the set $\ctd$ contains all possible ways of partitioning $X$ into $k$ subsets.
The constrained $k$-median problem can be defined similarly.
We will make the natural assumption (which is made by Ding and Xu as well) that it suffices to find a set of $k$ centers.
In other words, there is an (efficient) algorithm $\actd$, which given a set of $k$ centers $c_1, \ldots, c_k$, outputs the clustering $\{O_1, \ldots, O_k\} \in \ctd$ such that $\sum_{i=1}^k \sum_{x \in O_i} ||c_i - x||^2$ is minimized.
Such an algorithm is called a {\em partition algorithm} by Ding and Xu~\cite{dx15}
\footnote{\cite{dx15} also gave a discussion on such partition algorithms for a number of clustering problems with side constraints.}.
For the case of the $k$-means problem, this algorithm will just give the Voronoi partition with respect to $c_1, \ldots, c_k$, whereas in the case of the $r$-gather $k$-means clustering problem, the algorithm $\actd$ will be given by a suitable min-cost flow computation (see section 4.1 in \cite{dx15}).

Ding and Xu~\cite{dx15} considered several natural problems arising in diverse areas, e.g. machine learning, which can be stated in this framework.
These included the so-called $r$-gather $k$-means, $r$-capacity $k$-means and $l$-diversity $k$-means problems. Their approach for solving
such problems was to output a list of candidate sets of centers (of size $k$) such that at least one of these were close to the optimal centers. We
formalize this approach and show that if $k$ is a constant, then one can obtain a PTAS for the constrained $k$-means (and the constrained $k$-median)
problems whose running time is linear plus a constant number of calls to $\actd$.

We define the {\em list $k$-means} problem.
Given a set of points $X$ and  parameters $k$ and $\eps$, we want to output a list $\L$ of sets of $k$ points (or centers).
The list $\L$ should have the following property: for {\em any} partitioning $\bO=\{O_1, \ldots, O_k\}$ of $X$ into $k$ clusters, there exists a set $c_1, \ldots, c_k$ in the list $\L$ such that  (up-to reordering of these centers)
\begin{eqnarray}
\label{eq:list}
\sum_{i=1}^k \sum_{x \in O_i} ||c_i - x||^2 \leq (1+\eps)\sum_{i=1}^k \sum_{x \in O_i} ||x - m_i||^2,
 \end{eqnarray}
 where $m_i = \frac{\sum_{x \in O_i} x}{|O_i|}$ denotes the mean of $O_i$.
Note that the latter quantity is the $k$-means cost of the clustering $\bO$, and so we require $c_1, \ldots, c_k$ to be such that the
cost of assigning to these centers is close to the optimal $k$-means cost of this clustering.
We shall use $\opt_k(\bO)$ to denote the optimal $k$-means cost of $\bO$.
%% When $\bO$ is clear from the context, we use $\opt$ instead of $\opt_k(\bO)$.

Although such an oblivious approach to clustering may appear too optimistic, we show that it is possible to obtain such a list $\L$ of size $2^{\tilde{O}(k/\eps)}$ in $O \left(nd \cdot 2^{\tilde{O}(k/\eps)} \right)$ time.
This improves the result of Ding and Xu~\cite{dx15}, where they gave an algorithm which outputs a list of size $O \left((\log{n})^k \cdot 2^{\poly(k/\eps)}\right)$.
Observe that we address a question which is both algorithmic and existential~: how small can the size of $\L$ be, and how efficiently can we find it~?
We also give almost matching lower bounds on the size of such a list $\L$.
Our algorithm for finding $\L$ relies on the $D^2$-sampling idea --
iteratively find the centers by picking the next one to be {\em far} from the current set of centers.
Although these ideas have been used for the $k$-means problems (see e.g.~\cite{jks}), they rely heavily on the fact that given a set of centers, the corresponding clustering is obtained by the corresponding Voronoi partition.
Our approach relies in showing that there is small sized list $\L$ which works well for all possible clusterings.

It is not hard to show that a result for the list $k$-means problem implies a corresponding result for the constrained $k$-means problem with
the number of calls to $\actd$ being equal to the size of the list $\L$. Therefore, we obtain as corollary of our main result efficient algorithms
for the constrained $k$-means (and the constrained $k$-median) problems.

\subsection{Related work}

The classical $k$-means problem is one of the most well-studied clustering problems. There is a long sequence of work on obtaining
fast PTAS for the $k$-means and the $k$-median problems (see e.g., ~\cite{Matousek00,BadoiuHI02,VegaKKR03,Har-PeledM04,kss,abs10,Chen06,jks,FeldmanMS07}  and references therein). Some of these works implicitly maintain
 a list of centers of size $k$ such that the condition~(\ref{eq:list}) is satisfied for all clusterings $\bO$ which correspond to a Voronoi partition (with respect to a set of $k$ centers) of the input set of points, and one picks the best possible set of centers from this list (see e.g.,~\cite{kss,abs10,jks}).
 The list has at most $2^{\poly(k/\veps)}$ elements, and from this, one can recover a
 $(1 + \veps)$-approximation algorithm for the $k$-means problem with running time $O \left( n d \cdot 2^{\poly(k/\veps)}\right)$.

%
%
%For the classical $k$-means problem where the optimal clustering is a Voronoi partition of the dataset, some upper bounds have been known for sometime even though this problem was never posed in this manner (where a {\em list} of $k$-centers can be produced).
%This is due to the work done by Kumar, Sabharwal, and Sen~\cite{kss} and a number of follow-up works~\cite{abs10,jks}.
%The algorithm in all these works, implicitly produce a list of $2^{poly(k/\veps)}$ elements, each element being a set of $k$ centers, and then picks the one that minimises the objective function.
%This gives a $(1 + \veps)$-approximation algorithm for the $k$-means problem with running time $O \left( n d \cdot 2^{poly(k/\veps)}\right)$.
%A lot of work has been done in obtaining approximation algorithms for the $k$-means and $k$-median problems.
%The reader is advised to look at references~\cite{Matousek00,BadoiuHI02,VegaKKR03,Har-PeledM04,kss,abs10,Chen06,jks,FeldmanMS07} to obtain a high level picture of these works.
%
%\vspace{0.05in}
%
%\noindent

The more general case of the constrained $k$-means problem was studied by Ding and Xu~\cite{dx15} who also gave an algorithm that outputs a list of size $O \left( (\log{n})^{k} \cdot 2^{\poly(k/\veps)}\right)$. Our work improves upon this result.
%
%The case when the optimal clusters could be an arbitrary partition of the dataset was modelled by Ding and Xu~\cite{dx15} who also give an algorithm that outputs a list of size $O \left( (\log{n})^{k} \cdot 2^{poly(k/\veps)}\right)$.
%This work may also be regarded as improvement of their results.
Moreover, we consider the formulation of the list $k$-means problem as an important contribution, and feel that similar formulations in
other classification settings would be useful.

%
%we would like to point out
%that formulating the question as bounds on the list size and then showing a lower bound (in addition to the improved upper bound) are new. 

\subsection{Preliminaries}
We  formally define the problems considered in this paper.
The centroid or mean of a finite set of points $X \subset \mathbb{R}^d$ is denoted by $\Gamma(X) = \frac{\sum_{x \in X} x }{|X|}$.
Let $\Delta(X)$ denote the $1$-means cost of these set of points, i.e., $\sum_{x \in X} ||x-\Gamma(X)||^2$.

An input instance $\I$ for the list $k$-means (or the list $k$-median)  problem
consists of a set of points $X$, a positive integer  $k$ and a positive parameter $\eps$.
A partition of $X$ into disjoint subsets $O_1, \ldots, O_k$ will be called a {\em clustering} of $X$.
Given a clustering $\bOst = \{O_1^\star, \ldots, O_k^\star\}$ of $X$ and a set of $k$ centers $C=\{c_1, \ldots, c_k\}$, define $\cost_C(\bO^\star)$ as the minimum, over all permutations $\pi$ of $C$, of $\sum_{i=1}^k \sum_{x \in O_i^\star} ||x-c_{\pi(i)}||^2$.
Recall that $\opt_k(\bOst)$ denotes the optimal $k$-means cost of $\bOst$, i.e., $\sum_{i=1}^{k} \sum_{x \in O_i^\star}
||x - \Gamma(O_i^\star)||^2.$
%The goal
%%Let $\opt = \sum_{i=1}^{k} \sum_{x \in O_i^\star} ||x - \Gamma(O_i^\star)||^2$  denote the optimal cost w.r.t. clustering $\bO^\star$.
 %%is to output a list $L$ of sets of $k$ centers such that there is at least one set $C$ of $k$ centers in the list such that $\cost_C(\bO^\star) \leq (1 + \eps) \cdot \opt$.

For a set of points $X$ and a set of points $C$ (of size at most $k$), define $\Phi_{C}(X)$ as $\sum_{x \in X} \min_{c \in C} ||x - c||^2$, i.e., we consider the Voronoi partition of $X$ induced by $C$, and consider the $k$-means cost of $X$ with respect to this partition. When considering the list $k$-median problem, we will use the same notation, except that we will consider the Euclidean norm instead of the square of the Euclidean norm.
When $C$ is a singleton set $\{c\}$, we shall abuse notation by using $\Phi_c(X)$ instead of $\Phi_{\{c\}}(X)$.

As mentioned in the introduction, the constrained $k$-means problem is specified by a set of points $X$, a positive integer $k$, and
a set $\ctd$ of feasible clusterings of $X$. Further, we are given an algorithm $\actd$, which given a set of $k$ centers $C$, outputs
the clustering $\bO$ in $\ctd$ which minimizes $\cost_C(\bO)$. The goal is to find a clustering $\bO \in \ctd$ and a set $C$ of size $k$
which minimizes $\cost_C(\bO)$. Note that the centers in $C$ should just be the mean of each cluster in $\bO$. On the other hand, if we know $C$, then
we can find the best clustering in $\ctd$ by calling $\actd$. We use the same notation for the constrained $k$-median problem.

We now mention a few results which will be used in our analysis.
The following fact is well known.
\begin{fact}\label{lem:folklore}
For any $X \subset \mathbb{R}^d$ and $c \in \mathbb{R}^d$ we have
$\sum_{x \in X} ||x - c||^2 = \sum_{x \in X} ||x - \Gamma(X)||^2 + |X| \cdot ||c - \Gamma(X)||^2$.
\end{fact}
\noindent
We next define the notion of $D^2$-sampling.
\begin{definition}[$D^2$-sampling]
Given a set of points $X \subset \mathbb{R}^d$ and another set of points $C \subset \mathbb{R}^d$,  $D^2$-sampling from $X$ w.r.t. $C$
samples a point $x \in X$ with probability $\frac{\Phi_{C}(\{x\})}{\Phi_{C}(X)}$.
\end{definition}
\noindent
The following result of Inaba et al.~\cite{inaba} shows that a constant size random sample is a good enough approximation of a set of points
$X$ as far as the 1-means objective is concerned.
\begin{lemma}[\cite{inaba}]\label{lemma:inaba}
Let $S$ be a set of points obtained by independently  sampling $M$ points with replacement uniformly at random from a point set $X \subset \mathbb{R}^d$.
Then for any $\delta > 0$,
\[
\pr\left[\Phi_{\Gamma(S)}(X) \leq \left( 1 + \frac{1}{\delta M}\right) \cdot \Delta(X) \right] \geq (1- \delta).
\]
\end{lemma}
\noindent
%We will  use the following concentration result popularly known as Chernoff bounds.
%\begin{lemma}[Chernoff Bound]\label{lemma:chernoff}
%Let $X_1, ..., X_n$ be independent 0/1 random variables. Let $X = \sum_i X_i$ and $\mu = \sum_i \E[X_i]$. Let $\delta > 0$ be any real number. Then $\pr[X \leq (1 - \delta) \cdot \mu] \leq e^{-(\delta^2 \mu)/2}$.
%\end{lemma}
%\noindent
We will also use the following simple fact that may be interpreted as approximate version of the triangle inequality for squared Euclidean distance.

\begin{fact}[Approximate triangle inequality]\label{lem:triangle}
For any $x, y, z \in \mathbb{R}^d$, we have $||x - z||^2 \leq 2 \cdot ||x - y||^2 + 2 \cdot ||y - z||^2$.
\end{fact}

\subsection{Our results}
We now state our results for the list $k$-means and the list $k$-median problems.

\begin{theorem} \label{thm:upperbound}
Given a set of $n$ points $X \subset \mathbb{R}^d$, parameters $k$ and $\eps$, there is a randomized algorithm which outputs a list  $\L$ of $2^{\tilde{O}(k/\eps)}$ sets of centers of size $k$ such that for any clustering $\bOst = \{O_1^\star, ..., O_k^\star\}$ of $X$, the following event happens with probability at least $1/2$~: there is a set $C \in \L$ such that
$$
\cost_C(\bOst) \leq (1+\eps) \cdot \opt_k(\bOst).$$
%%Here, $\opt = \sum_{i=1}^{k} \sum_{x \in O_i^\star} ||x - \Gamma(O_i^\star)||^2$.
Moreover, the running time of our algorithm is $O \left(n d \cdot 2^{\tilde{O}(k/\eps)} \right)$.
The same statement holds for the list $k$-median problem as well, except that the size of the list $\L$ becomes $2^{\tilde{O}(k/\eps^{O(1)})}$
and the running time of our algorithm becomes $O \left(n d \cdot 2^{\tilde{O}(k/\eps^{O(1)})} \right)$.
\end{theorem}

\noindent
As a corollary of this result we get PTAS for the constrained $k$-means problem (and similarly for the constrained $k$-median problem). \sv{The proof
is deferred to the full version.}
\begin{corollary}
\label{cor:upperbound}
There is a randomized algorithm which given an instance of the constrained $k$-means problem and parameter $\eps > 0$, outputs a solution
of cost at most $(1+\eps)$-times the optimal cost with probability at least $1/2$.
 Further, the time taken by this algorithm is $O\left(n d \cdot 2^{\tilde{O}(k/\eps)} \right) +
2^{\tilde{O}(k/\eps)} \cdot T$, where $T$ denotes the time taken by $\actd$ on this instance.
\end{corollary}
\lv{
\begin{proof}
We use the algorithm in Theorem~\ref{thm:upperbound} to get a list $\L$ for this data-set. For each set $C \in \L$, we invoke $\actd$ with
$C$ as the set of centers -- let $\bO(C)$ denote the clustering produced by $\actd$. We output the clustering for which $\cost_C(\bO(C))$
is minimum. Let $\bOst$ be the optimal clustering, i.e., the clustering in $\ctd$ for which $\opt_k(\bOst)$ is minimum. We know that
with probability at least $1/2$, there is set $C \in \L$ for which $\cost_C(\bOst) \leq (1+\eps) \opt_k(\bOst)$. Now, the solution
produced by our algorithm has cost at most $\cost_C(\bO(C))$, which by definition of $\actd$, is at most $\cost_C(\bOst)$.
\qed
\end{proof}
}
\noindent
We also give a nearly matching lower bound on the size of $\L$. The following result along with Yao's Lemma shows that one cannot reduce the
size of $\L$ to less than $2^{\tilde{\Omega} \left(\frac{k}{\sqrt{\veps}} \right)}$.
\begin{theorem}
\label{thm:lowerbound}
 Given a parameter $k$ and a small enough positive constant $\eps$,
there exists a set $X$ of points in $\mathbb{R}^d$ and a set $\ctd$ of clusterings  of $X$ such that any list $\L$ of centers of size $k$
with the following property must have size at least $2^{\tilde{\Omega} \left(\frac{k}{\sqrt{\veps}} \right)}$~: for at least half of the
clusterings $\bO \in \ctd$, there exists a set $C$ in $\L$ such that $\cost_C(\bO) \leq (1+ \eps) \opt_k(\bO)$.
\end{theorem}

%\noindent
%We also show that the above results may be extended to the $k$-median setting.
Our techniques also extend to settings involving many other ``approximate" metric spaces \lv{(see the discussion in Section~\ref{sec:concl}).}
\sv{(see the full version for details.)}
Another important observation is that in the lower bound result above, the clusterings in $\ctd$ correspond to Voronoi partitions of $X$.
This throws light on the previous works~\cite{kss,abs10,FeldmanMS07,jks,jky15} as to why the running time of all the algorithms was proportional to $2^{\poly(k/\eps)}$: they were implicitly maintaining a list which satisfied~(\ref{eq:list}) for all Voronoi partitions of $X$, and therefore, our
lower bound result applies to their algorithms as well.

\subsection{Our Techniques}
Our techniques are based on the idea of $D^2$-sampling that was used by Jaiswal et al.~\cite{jks} to give a $(1+\veps)$-approximation algorithm for the $k$-means problem.
%%In fact, except for one small change, the algorithm in this paper is the same as that in~\cite{jks}.
Our ideas also have similarities to the ideas of Ding and Xu~\cite{dx15}.
We discuss these similarities towards the end of this subsection.

One of the crucial ingredients that is used in most of the $(1 + \veps)$-approximation algorithms for $k$-means is Lemma~\ref{lemma:inaba}.
This result essentially states that given a set of points $P$, if we are able to uniformly sample $O(1/\eps)$ points from it, then the mean of these sampled points will be a good substitute for the mean of $P$. Consider an optimal clustering $\ost{1}, \ldots, \ost{k}$ for a set of
points $X$. If we could uniformly sample from each of the clusters $\ost{i}$, then by the argument above, we will be done. The first problem
one encounters is that one can only sample from the input set of points, and so, if we sample sufficiently many points from $X$, we need
to somehow distinguish the points which belong to $\ost{i}$ in this sample.
This can be dealt with using the following argument: suppose we manage to get a small sample $S$ of points (say of size $O(\poly(k/\veps))$) that contain at least $\Omega(1/\veps)$ points uniformly distributed in $\ost{i}$, then we can try all possible subsets of $S$ of size $O(1/\veps)$ and ensure that at least one of the subsets is a uniform sample of appropriate size from $\ost{i}$.
Another issue is -- how do we ensure that the sample $S$  has sufficient representation from $\ost{i}$?
Uniform sampling from the input $X$ will not work since $|\ost{i}|$ might be really small compared to the size of $|X|$.
This is where $D^2$-sampling plays a crucial role and we discuss this next.

Given a set of points $X \subseteq \mathbb{R}^d$ and candidate centers $c_1, ..., c_i \in \mathbb{R}^d$, $D^2$-sampling with respect to the centers $c_1, ..., c_i$  samples a point $x \in X$ with probability proportional to $\min_{c \in \{c_1, ..., c_i\}} ||x - c||^2$.
Note that this process ``boosts" the probability of a cluster $\ost{j}$ that has many points far from the set  $\{c_1, \ldots, c_i\}$.
%So, a point sampled using $D^2$-sampling w.r.t. $c_1, ..., c_i$ will be more likely to be from a cluster that does not have representation within $c_1, ..., c_i$.
Therefore, even if a cluster $\ost{j}$ has a small size, we will have a good chance of sampling points from it (if it is far from the
current set of centers).
However, this nonuniform sampling technique gives rise to another issue.
The points being sampled are no longer  uniform samples from the optimal clusters.
Depending on the current set of centers, different points in a cluster $\ost{j}$ have different probability of getting sampled.
This issue is not that grave for the $k$-means problem where the optimal clusters are Voronoi regions since we can argue that the probabilities are not very different.
However, for the constrained $k$-means problem where the optimal clusters are allowed to be arbitrary partition of the input points, this problem becomes more serious.
This can be illustrated using the following example.
Suppose we have managed to pick centers $c_1, \ldots, c_i$ that are good (in terms of cluster cost) for the optimal clusters $O_1^\star, \ldots, O_i^\star$.
At this point let $\ost{j}$ denote the cluster other than $O_1^\star, \ldots , O_i^\star$, such that a point sampled using $D^2$ sampling w.r.t. $c_1, \ldots, c_i$ is most likely to be from $\ost{j}$.
%%We consider $O^\star$ just to ensure that $O^\star$ has a good representation in the $D^2$-sample taken at this point.
Suppose we sample a set $S$ of $O(k/\eps)$ points using $D^2$-sampling.
Are we guaranteed (w.h.p.) to have a subset in $S$ that is a uniform sample from $\ost{j}$?
The answer is no (actually quite far from it).
This is because the optimal clusters may form an arbitrary partition of the data-set and it is possible that most of the points in $\ost{j}$ might be very close to the centers $c_1, \ldots, c_i$.
In this case the probability of sampling such points will be close to $0$.
The way we deal with this scenario is that we consider a multi-set $S'$ that is the union of the set of samples $S$ and $O(1/\veps)$ copies of each of $c_1, \ldots, c_i$.
We then argue that all the points in $\ost{j}$ that is far from $c_1, \ldots, c_i$ will have a good chance of being represented in $S$ (and hence in $S'$).
On the other hand, even though the points that are close to one of $c_1, \ldots, c_i$ will not be represented in $S$ (and hence $S'$), the center
(among $c_1, \ldots, c_i$) that is close to these points have good representation in $S'$ and these centers may be regarded as ``proxy" for the points in $\ost{j}$.
%%The formal details are given in Section~\ref{upper-bound}.

Ding and Xu~\cite{dx15}, instead of using the idea of $D^2$-sampling, rely on the ideas of Kumar et al.~\cite{kss} which involves uniform sampling of points and then pruning the data-set by removing the points that are close to centers that are currently being considered.
In their work, they also encounter the  problem that  points from some optimal cluster might be close to the current set good centers (and hence will be removed before uniform sampling).
Ding and Xu~\cite{dx15} deal with this issue using what they call a ``simplex lemma".
Consider the same scenario as in the previous paragraph.
At a very high level, they consider  grids inside several  simplices defined by the current centers $c_1, \ldots, c_i$
and the sampled points.
%%and then uniformly sample points after removing the points that are close to $c_1,\ldots, c_i$.
Using the simplex lemma, they argue that one of the points inside these grids will be a good center for  the cluster $\ost{j}$.
%%So, our ideas have similarity to those in Ding and Xu~\cite{dx15} in the sense that our algorithm also takes into account the centers $c_1, ..., c_i$ when trying to find a good candidate center for $O^\star$.

We now give an overview of the paper. In Section~\ref{sec:algo}, we give the algorithm for generating the list of sets of centers for an instance
of the list $k$-means problem. The algorithm is analyzed in Section~\ref{sec:analysis}.
\lv{In Section~\ref{sec:lower}, we give the lower bound result on the size of the list $\L$. In Section~\ref{sec:ext}, we discuss how our algorithm can be extended to the list $k$-median problem. We conclude with a brief discussion on extensions to other metrics in Section~\ref{sec:concl}. }
\sv{Details about the lower bound construction (Theorem~\ref{thm:lowerbound} and extensions to the $k$-median problem, and other distance metric settings,  are discussed in the full version.}

%\input{finding-clusters}

%%\input{preliminaries}

%Upper bound for the k-means problem
%%\section{Upper Bound: Algorithm and Analysis}\label{upper-bound}

%%We start with the description of our algorithm in the next subsection.

\section{The Algorithm}
\label{sec:algo}
Consider an instance of the list $k$-means problem. Let $X$ denote the set of points, and $\eps$ be a positive parameter.
The algorithm {\bf List-$k$-means} is described in Figure~\ref{fig:k}. It maintains a set $C$ of centers, which is initially empty.
Each recursive call to the function {\bf Sample-centers} increases the size of $C$ by one.
In Step~2 of this function, the algorithm tries out various candidates
which can be added to $C$ (to increase its size by $1$).
First, it builds a multi-set $S$ as follows: it independently samples (with replacement) $O(k/\eps^3)$ points using $D^2$-sampling from $X$ w.r.t. the set $C$.
Further, it adds $O(1/\eps)$ copies of each of the centers in $C$ to the set $S$.
Having constructed $S$, we consider all subsets of size $O(1/\eps)$ of $S$ -- for each such subset we try adding the mean of this set to $C$.
Thus, each invocation of {\bf Sample-centers} makes multiple recursive calls to itself ($|S| \choose {M}$ to be precise).
It will be useful to think of the execution of this algorithm as a tree $\T$ of depth $k$.
Each node in the tree can be labeled with a set $C$ -- it corresponds to the invocation of {\bf Sample-centers} with this set as $C$ (and $i$ being the depth of this node).
The children of a node denote the recursive function calls by the corresponding invocation of {\bf Sample-centers}.
Finally, the leaves denote the set of candidate centers produced by the algorithm.
\begin{center}
\begin{Algorithm}
\begin{boxedminipage}{\textwidth}
{\bf List-$k$-means($X, k, \veps$)}

\hspace*{0.1in} - Let $N = \frac{136448 \cdot k}{\eps^3}$, $M = \frac{100}{\veps}$ \\
\hspace*{0.1in} - Initialize $\L$ to $\emptyset$. \\
\hspace*{0.1in} - Repeat $2^k$ times: \\
\hspace*{0.3in} - Make a call to {\bf Sample-centers$(X, k, \eps, 0, \{\})$}. \\
\hspace*{0.1in} - Return $\L$.
%%\vspace{0.1in}

{\bf Sample-centers$(X, k, \eps, i, C)$} \\
\sp \sp \sp \sp \sp \sp \sp \sp (1) If $(i = k)$ then add $C$ to the set $\L$. \\
\sp \sp \sp \sp \sp \sp \sp \sp (2) else \\
\sp \sp \sp \sp \sp \sp \sp \sp \sp \sp \sp \sp (a) Sample a multi-set $S$ of $N$ points with $D^2$-sampling (w.r.t. centers $C$) \\
\sp \sp \sp \sp \sp \sp \sp \sp \sp \sp \sp \sp  (b) $S' \leftarrow S$ \\
 \sp \sp \sp \sp \sp \sp \sp \sp \sp \sp \sp \sp  (c) For all $c \in C$: $S' \leftarrow S' \cup \{\textrm{$M$ copies of $c$}\}$ \\
 \sp \sp \sp \sp \sp \sp \sp \sp \sp \sp \sp \sp  (d) For all subsets $T \subset S'$ of size $M$: \\
 \sp \sp \sp \sp \sp \sp \sp \sp \sp \sp \sp \sp   \sp \sp \sp \sp (i) $C \leftarrow C \cup \{\Gamma(T)\}$. \\ %%\footnote{$\Gamma(T)$ denotes the centroid of the points in $T$, i.e., $\Gamma(T) = \frac{1}{|T|} \cdot \sum_{t \in T} t$.}\\
\sp \sp \sp \sp \sp \sp \sp \sp \sp \sp \sp \sp   \sp \sp \sp \sp  (ii) {\bf Sample-centers$(X, k, \eps, i+1, C)$}
\end{boxedminipage}
\caption{Algorithm for list $k$-means}
\label{fig:k}
\end{Algorithm}
\end{center}
\vspace*{-0.8in}

\section{Analysis}\label{sec:analysis}
In this section we prove Theorem~\ref{thm:upperbound} for the list $k$-means problem.
Let $\L$ denote the set of candidate solutions produced by {\bf List-$k$-means}, where
a solution corresponds to a set of centers $C$ of size $k$.
These solutions are output  at the leaves of the execution tree $\T$. Fix a clustering $\bOst=\{\ost{1}, \ldots, \ost{k}\}$ of
$X$.
%
%\begin{theorem}[Main Theorem]\label{thm:main}
%With at least $1/2$ probability, there is a solution $C = \{c_1, ..., c_k\}$ in the list $\L$ such that $ \cost_C(X) \leq (1+\eps) \cdot \opt$.
%Moreover, the algorithm {\bf Find-$k$-means} runs in time $O(n d \cdot 2^{\tilde{O}(k/\veps)})$.
%\end{theorem}
%
%\noindent
%We now give the proof of Theorem~\ref{thm:main}.
Recall that a node $v$ at depth $i$ in the execution tree $\T$ corresponds to a set $C$ of size $i$ -- call this set $C_v$.
Our proof will argue inductively that for each $i$, there will be a node $v$ at depth $i$ such that the centers chosen so far in $C_v$ are {\em good} with respect to a subset of $i$ clusters in $\ost{1}, \ldots, \ost{k}$.
We will argue that the following invariant $P(i)$ is maintained during the recursive calls to {\bf Sample-centers}:

\begin{quote}
{\boldmath${P(i)}$}: With probability at least $\frac{1}{2^{i-1}}$, there is a node $v_i$ at depth $(i-1)$ in the tree $\T$ and a set of $(i-1)$ distinct  clusters $\ost{j_1}, \ost{j_2}, ..., \ost{j_{i-1}}$ such that
\begin{equation}\label{eqn:invariant}
\forall l \in \{1, ..., i-1\}, \Phi_{c_l}(\ost{j_l}) \leq \left( 1 + \frac{\eps}{2} \right) \cdot \Delta(\ost{j_l}) + \frac{\eps}{2k} \cdot \opt_k(\bOst),
\end{equation}
where $c_1, \ldots, c_{i-1}$ are the centers in the set $C_{v_i}$ corresponding to $v_i$. Recall that $\Delta(\ost{j_l})$ refers to the optimal $1$-means cost of $\ost{j_l}$.
\end{quote}

The proof of the main theorem follows easily from this invariant property -- indeed, the statement $P(k)$ holds with probability at least $1/2^k$. Since
the algorithm {\bf List-$k$-means} invokes {\bf Sample-centers} $2^k$ times, the probability of the statement in $P(k)$ being true
in at least one of these invocations is at
least a constant. We now prove the invariant by induction on $i$.
The base case for $i=1$ follows trivially: the vertex $v_1$ is the root of the tree $\T$ and $C_{v_1}$ is empty.
Now assume that $P(i)$ holds for some $i \geq 1$.
We will prove that $P(i+1)$ also holds.
We first condition on the event in $P(i)$ (which happens with probability at least $\frac{1}{2^{i-1}}$).
Let $v_i$ and $\ost{j_1}, \ldots, \ost{j_{i-1}}$ be as guaranteed by the invariant $P(i)$.
Let $C_{v_i} = \{c_1, \ldots, c_{i-1}\}$ (as in the statement $P(i)$).
For sake of ease of notation, we assume without loss of generality that the index $j_i$ is $i$, and we shall use $C_i$ to denote $C_{v_i}$.
Thus, the center $c_l$ corresponds to the cluster $\ost{l}$, $1 \leq l \leq i-1$.
Note that for a cluster $\ost{i'}, i' \geq i$, $\Phi_{C_i}({\ost{i'}})$ is proportional to the probability that a point sampled from $X$ using $D^2$-sampling
w.r.t. $C_i$ comes from the set $\ost{i'}$ -- let $\ti \in \{i, \ldots, k\}$ be the index $i'$ for which $\Phi_{C_i}(\ost{i'})$ is maximum.
We will argue that the invocation of {\bf Sample-centers} corresponding to $v_i$ will try out a point $c_i$ (in Step 2(d)(i)) such that the following property will hold with probability at least $1/2$: $\Phi_{c_i}(\ost{\ti}) \leq (1 + \veps/2) \cdot \Delta(\ost{\ti}) + (\eps/2k) \cdot \opt_k(\bOst)$.
For doing this, we break the analysis into the following two parts.
These two parts are discussed in the next two subsections that follow.

\noindent
{\bf \underline{Case I}} $\left( \boldmath{\frac{\Phi_{C_i}(\ost{\ti})}{\sum_{j=1}^{k} \Phi_{C_i}(\ost{j})} < \frac{\eps}{13k}} \right)$: This captures the scenario where the probability of sampling from any of the uncovered clusters is very small.
Note that for the classical $k$-means problem, this is not an issue because in this case we can argue that the current set of centers $C$ already provides a good approximation for the entire set of data points and we are done.
However, for us this is an issue --- for example, assuming $i > 2$, it is possible that some of the points in $\ost{\ti}$ are close to $c_1$, whereas the remaining points of this cluster are close to $c_2$.
Still we need to output a center for $\ost{\ti}$.
In this case we argue that it will be sufficient to output a suitable convex combination of $c_1$ and $c_2$.

\noindent
{\bf \underline{Case II}} $\left(\boldmath{\frac{\Phi_{C_i}(\ost{\ti})}{\sum_{j=1}^{k} \Phi_{C_i}(\ost{j})} \geq \frac{\eps}{13k}} \right)$: In this case, we argue that with good probability we will sample  sufficient points from $\ost{\ti}$ during Step 2(a) of {\bf Sample-centers}.
Further, we will show that a suitable combination of such points along with centers in  $C_i$ will be a good center for $\ost{\ti}$.

\subsubsection{Case I $\left(\frac{\Phi_{C_i}(\ost{\ti})}{\sum_{j=1}^{k} \Phi_{C_i}(\ost{j})} < \frac{\eps}{13k} \right)$}:\\

In this case we argue that a convex combination of the centers in $C_i$  provides a good approximation to $\Delta(\ost{\ti})$.
Intuitively, this is because the points in $\ost{\ti}$ are close to the points in the set $C_i$.
This convex combination is essentially ``simulated" by taking $O(1/\eps)$ copies of each of the centers $c_1, ..., c_{i-1}$ in the multi-set $S$ and then trying all possible subsets of size $O(1/\eps)$.
The formal analysis follows.
First, we note that  $\Phi_{C_i}(\ost{\ti})$ should be small compared to $\opt_k(\bOst)$. \sv{The proof is deferred to the full version.}
\begin{lemma}\label{lem:inter0}
$\Phi_{C_i}(\ost{\ti}) \leq \frac{\eps}{6k} \cdot \opt_k(\bOst)$.
\end{lemma}
\lv{
\begin{proof}
Let $D$ denote $\sum_{j=1}^{k} \Phi_{C_i}(\ost{j})$.
The induction hypothesis and the fact that $\Phi_{C_i}(\ost{\ti}) \geq \Phi_{C_i}(\ost{j}), j \geq i,$ imply that
$$ D = \sum_{j=1}^{i-1} \Phi_{C_i}(\ost{j}) + \sum_{j=i}^k \Phi_{C_i}(\ost{j}) \leq \left( 1 + \frac{\eps}{2} \right) \cdot \sum_{j=1}^{i-1} \Delta(\ost{j}) + \frac{\eps}{2} \cdot \opt_k(\bOst) + k \cdot \Phi_{C_i}({\ost{\ti}}).$$
Since $\Phi_{C_i}(\ost{\ti}) \leq \frac{\eps}{13k} \cdot D$ and $\sum_{j=1}^{i-1} \Delta(\ost{j}) \leq \opt_k(\bOst)$, we get $D \leq \frac{\eps}{13} \cdot D + \left( 1 + \eps \right) \cdot \opt_k(\bOst).$
Thus, $D \leq \left( \frac{1+\eps}{1-\eps/13} \right) \cdot \opt_k(\bOst)$.
Finally, $\Phi_{C_i}(\ost{\ti}) \leq \frac{\eps}{13k} \cdot D \leq \frac{\eps}{6k} \cdot \opt_k(\bOst)$.
\qed
\end{proof}
}
For each point $p \in \ost{\ti}$, let $c(p)$ denote the closest center in $C_i$.
We now define a multi-set $\ostp{\ti}$ as $\{c(p) : p \in \ost{\ti} \}$.
Note that $\ostp{\ti}$ is obtained by taking multiple copies of points in $C_i$.
The remaining part of the proof proceeds in two steps.
Let $m^\star$ and $m'$ denote the mean of $\ost{\ti}$ and $\ostp{\ti}$ respectively.
We first show that $m^\star$ and $m'$ are close, and so, assigning all the points of $\ost{\ti}$ to $m'$ will have cost close to $\Delta(\ost{\ti})$.
Secondly, we show that if we have a good approximation $m''$ to $m'$, then assigning all the points of $\ost{\ti}$ to $m''$ will also incur small cost (comparable to $\Delta(\ost{\ti})$).
We now carry out these steps in detail.
Observe that
\begin{eqnarray}
\label{eq:dist1}
\sum_{p \in \ost{\ti}} ||p-c(p)||^2 = \Phi_{C_i}(\ost{\ti}).
\end{eqnarray}

\begin{lemma}\label{lem:mm'}
$||m^\star - m'||^2 \leq \frac{\Phi_{C_i}(\ost{\ti})}{|\ost{\ti}|}$.
\end{lemma}
\begin{proof}
Let $n$ denote $|\ost{\ti}|$. Then,
$$ ||m^\star - m'||^2 = \frac{1}{n^2} \left| \left| \sum_{p \in \ost{\ti}} (p-c(p)) \right|\right|^2 \leq \frac{1}{n} \sum_{p \in \ost{\ti}} ||p-c(p)||^2
= \frac{\Phi_{C_i}(\ost{\ti})}{n},$$
where the second last inequality follows from Cauchy-Schwartz
\footnote{For any real numbers $a_1, ..., a_m, (\sum_r a_r)^2/m \leq \sum_r a_r^2$.}.
\qed
\end{proof}

\noindent
Now we show that $\Delta(\ost{\ti})$ and $\Delta(\ostp{\ti})$ are close.

\begin{lemma}\label{lem:inter1}
$\Delta(\ostp{\ti}) \leq 2 \cdot \Phi_{C_i}(\ost{\ti}) + 2 \cdot \Delta(\ost{\ti})$.
\end{lemma}

\begin{proof}
The lemma follows by the following inequalities:
\begin{eqnarray*}
\Delta(\ostp{\ti}) & = & \sum_{p \in \ost{\ti}} ||c(p)-m'||^2
\stackrel{\mbox{\small{Fact~\ref{lem:folklore}}}}{\leq}   \sum_{p \in \ost{\ti}} ||c(p)-m^\star||^2   \\
& \stackrel{\mbox{\small{Fact~\ref{lem:triangle}}}}{\leq} & 2 \cdot \sum_{p \in \ost{\ti}}  \left( ||c(p)-p||^2 + ||p-m^\star||^2 \right)
 =  2 \cdot \Phi_{C_i}(\ost{\ti}) + 2 \cdot \Delta(\ost{\ti}).
\end{eqnarray*}
%%This completes the proof of the lemma.
\qed
\end{proof}

Finally, we argue that a good center for $\ostp{\ti}$ will also serve as a good center for $\ost{\ti}$.
\begin{lemma}\label{lem:inter2}
Let $m''$ be a point such that $\Phi_{m''}(\ostp{\ti}) \leq \left( 1+\frac{\eps}{8} \right) \cdot \Delta(\ostp{\ti})$.
Then $\Phi_{m''}(\ost{\ti}) \leq \left(1 + \frac{\eps}{2} \right) \cdot \Delta(\ost{\ti}) + \frac{\veps}{2k} \cdot \opt_k(\bOst)$.
\end{lemma}
\begin{proof}
Let $n^\star$ denote $|\ost{\ti}|$.
Observe that
\begin{eqnarray*}
\Phi_{m''}(\ost{\ti}) &=& \sum_{p \in \ost{\ti}} ||m''-p||^2 \stackrel{\mbox{\small{Fact~\ref{lem:folklore}}}}{=}
 \sum_{p \in \ost{\ti}} ||m^\star-p||^2 +  n^\star \cdot ||m^\star-m''||^2 \\
 & \hspace*{-1cm} \stackrel{\mbox{\small{Fact~\ref{lem:triangle}}}}{\leq} &
 \Delta(\ost{\ti}) + 2n^\star \left( ||m^\star - m'||^2 + ||m'-m''||^2 \right)  \stackrel{\mbox{\small{Lemma~\ref{lem:mm'}}}}{\leq}
 \Delta(\ost{\ti}) + 2 \cdot \Phi_{C_i}(\ost{\ti}) + 2 n^\star ||m'-m''||^2 \\
 & \hspace*{-1cm} \stackrel{\mbox{\small{Fact~\ref{lem:folklore}}}}{\leq} &
 \Delta(\ost{\ti}) + 2 \cdot \Phi_{C_i}(\ost{\ti}) + 2 \left( \Phi_{m''}(\ostp{\ti}) - \Delta(\ostp{\ti})\right)
 \ \leq \ \Delta(\ost{\ti}) + 2 \cdot \Phi_{C_i}(\ost{\ti}) + \frac{\eps}{4} \cdot \Delta(\ostp{\ti}) \\
& \hspace*{-1cm} \stackrel{\mbox{\small{Lemma~\ref{lem:inter1}}}}{\leq} & \Delta(\ost{\ti}) + 2 \cdot \Phi_{C_i}(\ost{\ti}) + \frac{\eps}{2} \cdot
 \left( \Phi_{C_i}(\ost{\ti}) + \Delta(\ost{\ti})\right)  \stackrel{\mbox{\small{Lemma~\ref{lem:inter0}}}}{\leq}
 \left(1 + \frac{\eps}{2} \right) \cdot \Delta(\ost{\ti}) + \frac{\veps}{2k} \cdot \opt_k(\bOst)
\end{eqnarray*}
This completes the proof of the lemma.
\qed
\end{proof}

The above lemma tells us that it will be sufficient to obtain a $(1 + \veps/8)$-approximation to the $1$-means problem for the dataset $\ostp{\ti}$.
Now, Lemma~\ref{lemma:inaba} tells us that there is a subset (again as a multi-set) $O''$
of size $\frac{16}{\veps}$ of $\ostp{\ti}$ such that the mean $m''$ of these points satisfies the conditions of Lemma~\ref{lem:inter2}.
Now, observe that $O''$ will be a subset of the set $S$ constructed in Step 2 of the algorithm {\bf Sample-center} -- indeed, in Step 2(c),
we add more than $\frac{16}{\eps}$ copies of {\em each} point in $C_i$ to $S$. Now, in Step 2(d),  we will try out all subsets of size
$\frac{16}{\eps}$ of $S$ and for each such subset, we will try adding its mean to $C_i$. In particular, there will be a recursive call of
this function, where we will have $C_{i+1}=C_i \cup \{ m''\}$ as the set of centers. Lemma~\ref{lem:inter2} now implies that $C_{i+1}$
will satisfy the invariant $P(i+1)$. Thus, we are done in this case.

\subsubsection{Case II $\left( \frac{\Phi_{C_i}(\ost{\ti})}{\sum_{j} \Phi_{C_i}(\ost{j})} \geq \frac{\eps}{13k} \right)$}:

\noindent
In this case, we would like to prove that we add a good approximation to the mean of $\ost{\ti}$ to the set $C_i$.
Again, consider the invocation of {\bf Sample-centers} corresponding to $C_i$.
We want the multi-set $S$ to contain a good representation from points in the set $\ost{\ti}$.
Secondly, in order to apply Lemma~\ref{lemma:inaba}, we will need this representation to be a uniform sample from $\ost{\ti}$.
Since $\Phi_{C_i}(\ost{\ti}) \geq \frac{\veps}{13k} \cdot \sum_j \Phi_{C_i}(\ost{j})$, the probability that a point sampled using $D^2$ sampling w.r.t. $C_i$ is from $\ost{\ti}$ is not too small.
So, the multi-set $S$ will have non-negligible representation from the set $\ost{\ti}$. However the points from $\ost{\ti}$ in $S$ may not be a uniform sample from $\ost{\ti}$.
Indeed, suppose there is a good fraction of points of $\ost{\ti}$ which are close to $C_i$, and remaining points of $\ost{\ti}$ are quite far from $C_i$.
Then, $D^2$-sampling w.r.t. to $C_i$ will not give us a uniform sample from $\ost{\ti}$.
To alleviate this problem, we take sufficiently many copies of points in $C_i$ and add them to the multi-set $S$.
In some sense, these copies act as proxy for points in $\ost{\ti}$ that are too close to $C_i$.
Finally, we argue that one of the subsets of $S$ ``simulates" a uniform sample from $\ost{\ti}$ and the mean of this subset provides a good approximation for the mean of $\ost{\ti}$.
The formal analysis follows.

We divide the points in $\ost{\ti}$ into two parts -- points which are close to a center in $C_i$, and the remaining points.
More formally, let the radius $R$ be given by

\begin{equation}\label{eqn:R}
R^2 = \frac{\veps^2}{41} \cdot \frac{\Phi_{C_i}(\ost{\ti})}{|\ost{\ti}|}
\end{equation}
Define $\ostn{\ti}$ as the points in $\ost{\ti}$ which are within distance $R$ of a center in $C_i$, and $\ostf{\ti}$ be the rest of the points in $\ost{\ti}$.
As in Case I, we define a new set $\ostp{\ti}$ where each point in $\ostn{\ti}$ is replaced by a copy of the corresponding point in $C_i$.
For a point $p \in \ostn{\ti}$, define $c(p)$ as the closest center in $C_i$ to $p$.
Now define a multi-set $\ostp{\ti}$ as $\ostf{\ti} \cup \{c(p): p \in \ostn{\ti}\}$.
Intuitively, $\ostp{\ti}$ denotes the set of points that are same as $\ost{\ti}$ except that points close to centers in $C_i$ have been ``collapsed" to these centers by taking appropriate number of copies.
Clearly, $|\ostp{\ti}| = |\ost{\ti}|$.
At a high level, we will argue that any center that provides a good $1$-means approximation for $\ostp{\ti}$ also provides a good approximation for
$\ost{\ti}$.
We will then focus on analyzing whether the invocation of {\bf Sample-centers} tries out  a good center for $\ostp{\ti}$.

We give some more notation. Let $m^\star$ and $m'$ denote the mean of $\ost{\ti}$ and $\ostp{\ti}$ respectively. Let $n^\star$ and $n$ denote
the size of the sets $\ost{\ti}$ and $\ostn{\ti}$ respectively. First, we show that $\Delta(\ost{\ti})$ is large with respect to $R$.

\begin{lemma}\label{lem:IIinter0}
 $\Delta(\ost{\ti}) = \Phi_{m^\star}(\ost{\ti}) \geq \frac{16 n}{\eps^2} R^2$.
\end{lemma}

\begin{proof}
Let $c$ be the center in $C_i$ which is closest to $m^\star$. We divide the proof into two cases:
\begin{itemize}
\item[(i)] $||m^\star - c|| \geq \frac{5}{\eps} \cdot R$: For any point $p \in \ostn{\ti}$, triangle inequality implies that
$$||p-m^\star|| \geq ||c(p)-m^\star|| - ||c(p)-p|| \geq  \frac{5}{\eps} \cdot R - R \geq  \frac{4}{\eps} \cdot R.$$
Therefore, \lv{$$ \Delta(\ost{\ti}) \geq \sum_{p \in \ostn{\ti}} ||p-m^\star||^2 \geq \frac{16 n}{\eps^2} R^2.$$}
\sv{$ \Delta(\ost{\ti}) \geq \sum_{p \in \ostn{\ti}} ||p-m^\star||^2 \geq \frac{16 n}{\eps^2} R^2.$}
\item[(ii)]  $ ||m^\star - c|| < \frac{5}{\eps} \cdot R $: In this case, we have
\begin{eqnarray*}
\Phi_{m^\star}(\ost{\ti}) & \stackrel{\mbox{\small{Fact~\ref{lem:folklore}}}}{=} &
\Phi_{c} (\ost{\ti}) - n^{\star} \cdot ||m^\star - c||^2 \geq \Phi_{C_i} (\ost{\ti}) - n^{\star} \cdot ||m^\star - c||^2
\\ &  \stackrel{(\ref{eqn:R})}{\geq}& \frac{41 n^\star }{\eps^2} \cdot R^2 - \frac{25 n^\star}{\eps^2} \cdot R^2 \geq \frac{16 n}{\eps^2} R^2.
\end{eqnarray*}
\end{itemize}
This completes the proof of the lemma.
\qed
\end{proof}

\noindent
\sv{The proofs of the following two lemmas are similar to those of Lemma~\ref{lem:mm'} and Lemma~\ref{lem:inter1} respectively, and are deferred to the full version.}
\begin{lemma}\label{lem:IImm'}
$||m^\star - m'||^2 \leq  \frac{n}{n^\star} \cdot R^2$
\end{lemma}
\lv{
\begin{proof}
Since the only difference between $\ost{\ti}$ and $\ostp{\ti}$ are the points in $\ostn{\ti}$, we get
$$||m^\star - m'||^2 = \frac{1}{(n^\star)^2} \left| \left| \sum_{p \in \ostn{\ti}} (p-c(p)) \right|\right|^2 \leq
\frac{n}{(n^\star)^2} \sum_{p \in \ostn{\ti}} ||p-c(p)||^2 \leq \frac{n^2}{(n^\star)^2} R^2 \leq \frac{n}{n^\star} \cdot R^2.$$
where the first inequality follows from the Cauchy-Schwartz inequality.
\qed
\end{proof}
}

\lv{
\noindent
We now show that  $\Delta(\ostp{\ti})$ is close to $\Delta(\ost{\ti})$.}
\begin{lemma}\label{lem:IIinter1}
$\Delta(\ostp{\ti}) \leq 4  n  R^2 + 2 \cdot \Delta(\ost{\ti})$.
\end{lemma}
\lv{
\begin{proof}
The lemma follows from the following sequence of inequalities:
\begin{eqnarray*}
\Delta(\ostp{\ti}) & = & \sum_{p \in \ostn{\ti}} ||c(p)-m'||^2  + \sum_{p \in \ostf{\ti}} ||p-m'||^2 \\
& \stackrel{\mbox{\small{Fact~\ref{lem:triangle}}}}{\leq} & \sum_{p \in \ostn{\ti}} 2(||c(p)-p||^2 + ||p-m'||^2) + \sum_{p \in \ostf{\ti}} ||p-m'||^2\\
& \leq & 2nR^2 + 2 \sum_{p \in \ost{\ti}} ||p-m'||^2 = 2nR^2 + 2 \cdot \Phi_{m'}(\ost{\ti}) \\
& \stackrel{\mbox{\small{Fact~\ref{lem:folklore}}}}{=} & 2nR^2 + 2 \cdot \left( \Delta(\ost{\ti}) + n^\star \cdot ||m' - m^{\star}||^2\right) \\
& \stackrel{\mbox{\small{Lemma~\ref{lem:IImm'}}}}{\leq} & 4nR^2 + 2 \cdot \Delta(\ost{\ti}).
\end{eqnarray*}
This completes the proof of the lemma.
\qed
\end{proof}

}

\noindent
We now argue that any center that is good for $\ostp{\ti}$ is also good for $\ost{\ti}$.
\begin{lemma}
\label{lem:II-final}
Let $m''$ be  such that $\Phi_{m''}(\ostp{\ti}) \leq \left(1 + \frac{\eps}{16} \right) \cdot \Delta(\ostp{\ti})$.
Then $\Phi_{m''}(\ost{\ti}) \leq \left(1 + \frac{\eps}{2} \right) \cdot \Delta(\ost{\ti}) .$
\end{lemma}

\begin{proof}
The lemma follows from the following inequalities:
\begin{eqnarray*}
\Phi_{m''}(\ost{\ti}) &=& \sum_{p \in \ost{\ti}} ||m''-p||^2 \stackrel{\mbox{\small{Fact~\ref{lem:folklore}}}}{=}
 \sum_{p \in \ost{\ti}} ||m^\star-p||^2 +  n^\star \cdot ||m^\star-m''||^2 \\
& \hspace*{-1cm} \stackrel{\mbox{\small{Fact~\ref{lem:triangle}}}}{\leq} &
 \Delta(\ost{\ti}) + 2n^\star \left( ||m^\star - m'||^2 + ||m'-m''||^2 \right)  \stackrel{\mbox{\small{Lemma~\ref{lem:IImm'}}}}{\leq}
 \Delta(\ost{\ti}) + 2n R^2 + 2n^\star \cdot ||m'-m''||^2 \\
 & \hspace*{-1cm} \stackrel{\mbox{\small{Fact~\ref{lem:folklore}}}}{\leq} &
 \Delta(\ost{\ti}) + 2nR^2 + 2 \cdot \left( \Phi_{m''}(\ostp{\ti}) - \Delta(\ostp{\ti})\right)
 \ \leq \ \Delta(\ost{\ti}) + 2 nR^2 + \frac{\eps}{8} \cdot \Delta(\ostp{\ti}) \\
& \hspace*{-1cm} \stackrel{\mbox{\small{Lemma~\ref{lem:IIinter1}}}}{\leq} & \Delta(\ost{\ti}) + 2 nR^2 + \frac{\eps}{2} \cdot nR^2 + \frac{\eps}{4} \cdot \Delta(\ost{\ti})
  \stackrel{\mbox{\small{Lemma~\ref{lem:IIinter0}}}}{\leq} \left(1 + \frac{\eps}{2} \right) \cdot \Delta(\ost{\ti}).
\end{eqnarray*}
This completes the proof of the lemma.
\qed
\end{proof}

Given the above lemma, all we need to argue is that our algorithm indeed considers a center $m''$ such that $\Phi_{m''}(\ostp{\ti}) \leq (1+\veps/16) \cdot \Delta(\ostp{\ti})$. For this we would need about $O(1/\eps)$ uniform samples from $\ostp{\ti}$. However, our algorithm can only sample using $D^2$-sampling w.r.t. $C_i$. For ease of notation, let $\ostpn{\ti}$ denote the multi-set $\{c(p): p \in \ostn{\ti}\}$.
Recall that $\ostp{\ti}$ consists of $\ostf{\ti}$ and $\ostpn{\ti}$.
The first observation is that the probability of sampling an element from $\ostf{\ti}$ is reasonably large (proportional to $\eps/k$). Using this fact, we
show how to sample from $\ostp{\ti}$ (almost uniformly). Finally, we show how to convert this almost uniform sampling to uniform sampling (at the cost
of increasing the size of sample). \sv{We defer the proof of the following lemma to the full version.}

\begin{lemma}
\label{lem:osample}
Let $x$ be a sample from $D^2$-sampling w.r.t. $C_i$.
Then, $\pr[x \in \ostf{\ti}] \geq \frac{\eps}{15k}$.
Further, for any point $p \in \ostf{\ti}$, $\pr[x=p] \geq \frac{\gamma}{|\ost{\ti}|}$, where $\gamma$ denotes $\frac{\veps^2}{533 k}$.
\end{lemma}
\lv{
\begin{proof}
Note that $\sum_{p \in \ost{\ti} \setminus \ostf{\ti}} \pr[x=p] \leq \frac{R^2}{\Phi_{C_i}(X)} \cdot |\ost{\ti}| \leq \frac{\eps^2}{41} \frac{\Phi_{C_i}(\ost{\ti})} {\Phi_{C_i}(X)}$.
Therefore, the fact that we are in case~II implies that
$$\pr[x \in \ostf{\ti}] \geq \Pr[x \in \ost{\ti}] - \pr[x \in \ost{\ti} \setminus \ostf{\ti}] \geq \frac{\Phi_{C_i}(\ost{\ti})}{\Phi_{C_i}(X)} - \frac{\eps^2}{41} \frac{\Phi_{C_i}(\ost{\ti})} {\Phi_{C_i}(X)} \geq \frac{\eps}{15k}.$$

\noindent
Also, if $x \in \ostf{\ti}$, then $\Phi_{C_i}(\{x\}) \geq R^2=\frac{\eps^2}{41} \cdot \frac{\Phi_{C_i}(\ost{\ti})}{|\ost{\ti}|}$.
Therefore,
$$\frac{\Phi_{C_i}(\{x\})}{\Phi_{C_i}(X)}  \geq \frac{\eps}{13k} \cdot \frac{R^2}{\Phi_{C_i}(\ost{\ti})} \geq \frac{\veps}{13k} \cdot \frac{\veps^2}{41} \cdot \frac{1}{|\ost{\ti}|} \geq \frac{\veps^2}{533 k} \cdot \frac{1}{|\ost{\ti}|}.
$$
This completes the proof of the lemma.
\qed
\end{proof}
}

Let $X_1, \ldots X_l$ be $l$ points sampled independently using $D^2$-sampling w.r.t. $C_i$.
We construct a new set of random variables $Y_1, \ldots, Y_l$.
Each variable $Y_u$ will depend on $X_u$ only, and will take values either in $\ostp{\ti}$ or will be $\nl$.
These variables are defined as follows: if $X_u \notin \ostf{\ti}$, we set $Y_u$ to  $\nl$. Otherwise, we assign $Y_u$ to one of the following random variables with equal probability:
(i) $X_u$ or (ii) a random element of the multi-set $\ostpn{\ti}$.
The following observation follows from Lemma~\ref{lem:osample}\lv{.}\sv{, and its proof is deferred to the full version.}

\begin{corollary}
\label{cor:osample}
For a fixed index $u$, and an element $x \in \ostp{\ti}$, $\pr[Y_u=x] \geq \frac{\gamma'}{|\ostp{\ti}|},$ where $\gamma'=\gamma/2$.
\end{corollary}
\lv{
\begin{proof}
If $x \in \ostf{\ti}$, then we know from Lemma~\ref{lem:osample} that $X_u$ is $x$ with probability at least $\frac{\gamma}{|\ostp{\ti}|}$ (note that
$\ostp{\ti}$ and $\ost{\ti}$ have the same cardinality). Conditioned on this event, $Y_u$ will be equal to $X_u$ with probability $1/2$.
Now suppose $x \in \ostpn{\ti}$. Lemma~\ref{lem:osample} implies that $X_u$ is an element of $\ostf{\ti}$ with probability at least $\frac{\eps}{15k}$.
Conditioned on this event, $Y_u$ will be equal to $x$ with probability at least $\frac{1}{2} \cdot \frac{1}{|\ostpn{\ti}|}$. Therefore,
the probability that $X_u$ is equal to $x$ is at least $\frac{\eps}{15k} \cdot \frac{1}{2|\ostpn{\ti}|} \geq \frac{\eps}{30k |\ostp{\ti}|} \geq \frac{\gamma'}{|\ostp{\ti}|}$.
\qed
\end{proof}
}

Corollary~\ref{cor:osample} shows that we can obtain samples from $\ostp{\ti}$ which are nearly uniform (up to a constant factor).
To convert this to a set of uniform samples, we use the idea of~\cite{jks}.
For an element $x \in \ostp{\ti}$, let $\gamma_x$ be such that $\frac{\gamma_x}{|\ostp{\ti}|}$ denotes the probability that the random variable $Y_u$ is equal to $x$ (note that this is independent of $u$).
Corollary~\ref{cor:osample} implies that $\gamma_x \geq \gamma'$.
We define a new set of independent random variables $Z_1, \ldots, Z_l$.
The random variable $Z_u$ will depend on $Y_u$ only.
If $Y_u$ is $\nl$, $Z_u$ is also $\nl$.
If $Y_u$ is equal to $x \in \ostp{\ti}$, then $Z_u$ takes the value $x$ with probability $\frac{\gamma'}{\gamma_x}$, and $\nl$ with the remaining probability.
\lv{Note that $Z_u$ is either $\nl$ or one of the elements of $\ostp{\ti}$.
Further, conditioned on the latter event, it is a uniform sample from $\ostp{\ti}$.}
We can now prove the key lemma\lv{.}\sv{(proof is deferred to the full version).}

\begin{lemma}
\label{lem:final}
Let $l$ be $\frac{128}{\gamma' \cdot \eps}$, and $m''$ denote the mean of the non-null samples from $Z_1, \ldots, Z_l$. Then, with probability at least $1/2$,
$\Phi_{m''}(\ostp{\ti}) \leq (1+\eps/16) \cdot \Delta(\ostp{\ti})$.
\end{lemma}
\lv{
\begin{proof}
Note that a random variable $Z_u$ is equal to a specific element of $\ostp{\ti}$ with probability equal to $\frac{\gamma'}{|\ostp{\ti}|}$.
Therefore, it takes $\nl$ value with probability $1-\gamma'$.
Now consider a different set of iid random variables $Z_u'$, $1 \leq u \leq l$ as follows: each $Z_u$ tosses a coin with probability of Heads being $\gamma'$.
If we get Heads, it gets value $\nl$, otherwise it is equal to a random element of $\ostp{\ti}$. It is easy to check that the joint distribution of the random variables $Z_u'$ is identical to that
of the random variables $Z_u$.
Thus, it suffices to prove the statement of the lemma for the random variables $Z_u'$.

Now we condition on the coin tosses of the random variables $Z_u'$.
Let $n'$ be the number of the number of random variables which are not $\nl$.
($n'$ is a deterministic quantity because we have conditioned on the coin tosses).
Let $m''$ be the mean of such non-$\nl$ variables among $Z_1', \ldots, Z_l'$.
If $m''$ happens to be larger than $64/\eps$, Lemma~\ref{lemma:inaba} implies that with probability at least $3/4$,
$\Phi_{m''}(\ostp{\ti}) \leq (1+\eps/16) \cdot \Delta(\ostp{\ti})$.

Finally, observe that the expected number of non-$\nl$ random variables is $\gamma' \cdot l \geq 128/\eps$.
Therefore, with probability at least $3/4$, the number of non-$\nl$ elements will be at least $64/\eps$.
\qed
\end{proof}
}

Let $C_i^{(l)}$ denote the multi-set obtained by taking $l$ copies of each of the centers in $C_i$. Now observe that all the non-$\nl$
elements among $Y_1, \ldots, Y_l$ are elements of $\{X_1, \ldots, X_l\} \cup C_i^{(l)}$, and so the same must hold for
$Z_1, \ldots, Z_l$. This implies that in Step 2(d) of the algorithm {\bf Sample-centers}, we would have tried adding the point $m''$ as
described in Lemma~\ref{lem:final}. Therefore, the induction hypothesis continues to hold with probability at least 1/2.
This concludes the proof of Theorem~\ref{thm:upperbound}.

\lv{
% Lower bound for the k-means problem
\section{Lower Bound}
\label{sec:lower}
\newcommand{\coord}[2]{(#1)_{#2}}

In this section, we prove the lower bound result Theorem~\ref{thm:lowerbound}. Consider parameters $k$ and $\eps$ (assume $\eps$ is a
small enough constant).
We first define the set of points $X$. Let $m$ denote  $\lceil \frac{1}{\sqrt{\veps}}\rceil$. The points will belong to
$\mathbb{R}^d$, where $d = km$. The set $X$ will have $d$ points, namely, $e_1, \ldots, e_d$, where $e_i$ denotes the vector which
has all coordinates 0, except for the $i^{th}$ coordinate, which is 1. Now, we define the set $\ctd$ of clusterings of $X$. The set $\ctd$
will consist of those clusterings $\bO = \{O_1, \ldots, O_k\}$ for which each of the clusters has exactly $m$ points. Observe that
\begin{eqnarray}
\label{eq:count}
|\ctd| = \frac{(km)!}{(m!)^{k}}
\end{eqnarray}

Now fix a set $C$ of $k$ centers, $c_1, \ldots, c_k$. We will now upper bound the number of clusterings $\bO \in \ctd$ for which
\begin{eqnarray}
\label{eq:lower1}
\cost_C(\bO) \leq (1+\eps) \opt_k(\bO).
\end{eqnarray}
Let $\bO=\{O_1, \ldots, O_k\}$ be as above. Note that
\begin{eqnarray}
\label{eq:count1}
\opt_k(\bO)  = \sum_{i=1}^k \Delta(O_i) = km \cdot \left( (1 - 1/m)^2 + (m-1) \cdot 1/m^2\right) = k(m-1)
\end{eqnarray}
Recall that $\cost_C(\bO)$ is obtained by assigning each cluster in $\bO$ to a unique center in $C$, and then by computing the sum of square
of distances of points in $X$ to the corresponding centers. Wlog we rearrange the clusters in $\bO$ such  that the points in $O_j$ are assigned to
$c_j$. For a vector $v$, we shall use $\coord{v}{j}$ to denote the $j^{th}$ coordinate of $v$.
For every center $c_r$ we define a corresponding vector $v_r$ as follows:
$$\coord{v_r}{j} =
\left\{
\begin{array}{cc} \coord{c_r}{j} & {\mbox{ if $e_j \notin O_r$}} \\
 \coord{c_r}{j} - \frac{1}{m} &
{\mbox{ otherwise}}  \end{array}
\right.$$
\begin{lemma}
\label{lem:lower1}
$\sum_{r=1}^k ||v_r||^2 \leq \frac{k}{m(m-1)}$.
\end{lemma}
\begin{proof}
Fix a cluster $O_r$. Let $m_r$ denote the mean of $O_r$. Note that $\coord{m_r}{j}$ is $1/m$ if $e_j \in O_r$, $0$ otherwise. We now
simplify the expression $\cost_C(\bO)$ as follows:
\begin{eqnarray*}
\cost_C(\bO) & = & \sum_{r=1}^k \sum_{e_j \in O_r} ||e_j - c_r||^2 \stackrel{\mbox{\small{Fact~\ref{lem:folklore}}}}{=}
\sum_{r=1}^k \sum_{e_j \in O_r} \left( ||e_j - m_r||^2 + ||m_r - c_r||^2 \right) \\
& = & \opt_k(\bO) + \sum_{r=1}^k m \cdot ||m_r - c_r||^2 = \opt_k(\bO) + m \sum_{r=1}^k  ||v_r||^2
\end{eqnarray*}
By our assumption, $\cost_C(\bO) \leq (1+\eps) \opt_k(\bO)$. Therefore,
$$ \sum_{r=1}^k ||v_r||^2 \leq \frac{\eps}{m} \cdot \opt_k(\bO) \stackrel{(\ref{eq:count1})}{=} \frac{\eps}{m} \cdot k (m-1) \leq \frac{k}{m(m-1)}.$$
\qed
\end{proof}

Now define a corresponding assignment function $f: X \rightarrow \{1, \ldots, k\}$ as follows: $f(e_j) = r$ if $e_j \in O_r$.  Let
$\bO'=\{O_1', \ldots, O_k'\}$ be another clustering in $\ctd$ which satisfies condition~(\ref{eq:lower1}). Define vectors $v_r'$ and the assignment function
$f'$ in a similar manner. The following lemma shows that $f$ and $f'$ cannot differ in too many coordinates.
\begin{lemma}
\label{lem:lower2}
Let $D$ denote the set of indices $j$ for which $f(e_j) \neq f'(e_j)$. Then $|D| \leq d/2$.
\end{lemma}
\begin{proof}
Assume for the sake of contradiction that $|D| > d/2$.
For cluster $O_r$, let $D_r$ denote the set of indices $j$ such that $e_j \in O_r \triangle O_r'$. Observe that $\coord{v_r}{j}$ and
$\coord{v_r'}{j}$ differ (in absolute value) by $1/m$. Therefore,
$$ ||v_r'||^2 = \sum_{j \in D_r} \left( \coord{v_r}{j} \pm \frac{1}{m} \right)^2
\geq \frac{|D_r|}{m^2} - \frac{2}{m} \sum_{j \in D_r} |\coord{v_r}{j}| .$$
Summing over $r = 1, \ldots, k$, we get
$$\sum_{r=1}^k ||v_r'||^2 \geq \frac{2|D|}{m^2} - \frac{2}{m} \sum_{r=1}^k \sum_{j \in D_r}  |\coord{v_r}{j}|
\geq \frac{d}{m^2} - \frac{2}{m} \cdot \sqrt{2d} \cdot \sqrt{\sum_{r=1}^k \sum_{j \in D_r}  |\coord{v_r}{j}|^2},$$
where the last inequality follows from Cauchy-Schwarz, and the observation that $\sum_r |D_r| = 2|D| > d$. Using Lemma~\ref{lem:lower1},
we see that
$$\sum_{r=1}^k ||v_r'||^2  \geq \frac{k}{m} - \frac{2}{m} \cdot \sqrt{2km} \cdot \sqrt{\sum_{r=1}^k ||v_r||^2} \geq
 \frac{k}{m} -  \frac{4k}{m \sqrt{m-1}} > \frac{k}{m(m-1)},$$
 assuming $m$ is a large enough constant. But this contradicts Lemma~\ref{lem:lower1}.
 \qed
\end{proof}

The above lemma shows that the number of clusterings in $\ctd$ satisfying condition~(\ref{eq:lower1}) is  small.
\begin{corollary}
\label{cor:lower}
The number of clusterings in $\ctd$ satisfying condition~(\ref{eq:lower1}) is at most $\binom{km}{km/2} \cdot \frac{(km/2)!}{((m/2)!)^k}$.
\end{corollary}

\begin{proof}
Fix a clustering $\bO=\{O_1, \ldots, O_r\}$
 satisfying condition~(\ref{eq:lower1}), and let $f$ be the corresponding assignment function. How many assignment
functions (corresponding to a clustering in $\ctd$) can differ from $f$ in at most $d/2$ coordinates ? There are at most $\binom{km}{km/2}$
ways of choosing the coordinates in which the two functions differ. Consider a fixed choice of such coordinates, and say there are $d_r$ coordinates
corresponding to points in $O_r$. Let $d'$ denote $\sum_r d_r$ (and so, $d' \leq d/2$). Now, we need to partition these coordinates
into sets of size $d_1, \ldots, d_k$ (note that $f'$ corresponds to a clustering where all clusters are of equal size). The number of possibilities here is
$\frac{d'!}{d_1! \ldots d_k!}$, which is at most $\frac{(d/2)!}{(d/2k)!)^k}.$
\qed
\end{proof}

Recall that we want $\L$ to contain enough elements such that for at least half of the clusterings in $\ctd$, condition~(\ref{eq:lower1})  is
satisfied with respect to some set of centers in $\L$. Therefore, Corollary~\ref{cor:lower} and (\ref{eq:count}) imply that
$$ |\L| \geq \frac{\frac{(km)!}{(m!)^k}}{\binom{km}{km/2} \cdot \frac{(km/2)!}{((m/2)!)^k}} = 2^{\tilde{\Omega}(km)} = 2^{\tilde{\Omega}(k/\sqrt{\veps})}.$$

This concludes the proof of Theorem~\ref{thm:lowerbound}.

% Extension to the k-median problem
\section{Extension to the list $k$-median problem}
\label{sec:ext}

The setting for the list $k$-median problem is same as that for the list $k$-means problem, except for the fact that distances are measured using the Euclidean norm (instead of the square of the Euclidean norm).
As before, for a set $C$ of  $k$ centers, and a clustering $\bO=\{O_1, \ldots, O_k\}$ of a set of points $X$, define $\cost_C(\bO)$ as the minimum, over all permutations $\pi$ of $C$, of $\sum_{i=1}^k \sum_{x \in O_i} ||x-c_{\pi(i)}||.$
Define  $\opt_k(\bO)$, $\Phi_C(X)$  analogously.
%%Note that $\cst{i}$ need not be the mean of $\ost{i}$.
For a set of points $X$, let $\Delta(X)$ denote the optimal $1$-median cost of $X$, i.e., $\min_{c \in {\mathbb{R}^d}} \sum_{x \in X} ||x-c||$.
We no longer have an analogue of Fact~\ref{lem:folklore} -- for a set of points $X$, if $c^\star$ denotes the optimal center with respect to the
$1$-median objective, and  $c$ is a point such that $\Phi_c(X) \leq (1+\eps) \cdot \Phi_{c^\star}(X)$, it is possible that $||c-c^\star||$ is large.
This also implies that there is no analogue of the Lemma~\ref{lemma:inaba}.
However, instead of the approximate triangle inequality~(Fact~\ref{lem:triangle}), we get triangle inequality in the Euclidean metric.

We shall use a result of Kumar  et al.~\cite{kss}, which gives an alternative to Lemma~\ref{lemma:inaba}, although it outputs several candidate centers instead of just the mean of a random sample.
\begin{lemma}[Theorem 5.4~\cite{kss}]
\label{lem:kss}
Given a random sample (with replacement) $R$ of size $\frac{1}{\eps^4}$ from a set of points $X \in {\mathbb{R}}^d$, there is a procedure
$\const(R)$, which outputs a set $\core(R)$ of size $2^{\left(1/\eps\right)^{O(1)}}$ such that the following event happens with probability at least
$1/2$~: there is at least one point $c \in \core(R)$ such that $\Phi_c(X) \leq (1+\eps) \cdot \Delta(X)$.
The time taken by the procedure $\const(R)$ is $O \left(2^{\left(1/\eps\right)^{O(1)}} \cdot d \right)$.
\end{lemma}

Now we explain the changes needed in the algorithm and the analysis.
Given a set of points $X$ and another set of points $C$, $D$-sampling from $X$ w.r.t. $C$ samples a point $x \in X$ with probability proportional to $\Phi_C(x)$, i.e., $\min_{c \in C} ||c-x||$.

\subsection{The algorithm}
The algorithm is the same as that in Figure~\ref{fig:k}, except for some minor changes in the procedure {\bf Sample-Centers}, 
and changes in the values of the various parameters. The parameters $\alpha$ and $\beta$ in the procedure {\bf List-$k$-median}
are large enough constants. We briefly describe the changes in the procedure {\bf Sample-Centers}.
In Step~2(a), we sample the multi-set $S$ using $D$-sampling w.r.t $C$.
We replace Step~2(d) by the following: for all subsets $T \subset S'$ of size $M$, and for all elements $c \in  \core(T)$ (i) $C \leftarrow C \cup \{c\}$, (ii) {\bf Sample-centers$(X, k, \eps, i+1, C)$}.
Recall that $\core(T)$ is the set guaranteed by Lemma~\ref{lem:kss}.
In other words, unlike for the $k$-means setting, where we could just work with the mean of $T$, we now need to try out all the elements in $\core(T)$.
Figure~\ref{fig:k-median}, gives a detailed description of the algorithm.

\begin{center}
\begin{Algorithm}
\begin{boxedminipage}{\textwidth}
{\bf List-$k$-median($X, k, \veps$)}

\hspace{0.1in} - Let $N = \frac{\alpha \cdot k}{\eps^6}$, $M = \frac{\beta}{\veps^4}$, $\L \leftarrow \emptyset$.

\hspace{0.1in} - Repeat $2^k$ times:

\hspace{0.3in} - Make a call to {\bf Sample-centers$(X, k, \eps, 0, \{\})$} and output the union of lists returned by these calls.

\hspace{0.1in} - Return $\L$.

\vspace{0.1in}

{\bf Sample-centers$(X, k, \eps, i, C)$} \\
\sp \sp \sp \sp \sp \sp \sp \sp (1) If $(i = k)$ then add $C$ to $\L$. \\
\sp \sp \sp \sp \sp \sp \sp \sp (2) else \\
\sp \sp \sp \sp \sp \sp \sp \sp \sp \sp \sp \sp (a) Sample a multiset $S$ of $N$ points with $D$-sampling (w.r.t. centers $C$) \\
\sp \sp \sp \sp \sp \sp \sp \sp \sp \sp \sp \sp  (b) $S' \leftarrow S$ \\
 \sp \sp \sp \sp \sp \sp \sp \sp \sp \sp \sp \sp  (c) For all $c \in C$: $S' \leftarrow S' \cup \{\textrm{$M$ copies of $c$}\}$ \\
 \sp \sp \sp \sp \sp \sp \sp \sp \sp \sp \sp \sp  (d) For all subsets $T \subset S'$ of size $M$ and for all elements $c \in \core(T)$: \\
 \sp \sp \sp \sp \sp \sp \sp \sp \sp \sp \sp \sp   \sp \sp \sp \sp (i) $C \leftarrow C \cup \{c\}$. \\
\sp \sp \sp \sp \sp \sp \sp \sp \sp \sp \sp \sp   \sp \sp \sp \sp  (ii) {\bf Sample-centers$(X, k, \eps, i+1, C)$}
\end{boxedminipage}
\caption{Algorithm for list $k$-median.}
\label{fig:k-median}
\end{Algorithm}
\end{center}
\vspace{-0.3in}

\subsection{Analysis}
The analysis proceeds along the same lines as in Section~\ref{sec:analysis}, and we would again like to prove the induction hypothesis
$P(i)$.
We use the same notation as in Section~\ref{sec:analysis}, and define Cases I and II analogously.
Consider Case I first.
Proof of Lemma~\ref{lem:inter0} remains unchanged.
The set $\ostp{\ti}$ is defined similarly.
Let $m^\star$ be the point for which $\Delta(\ost{\ti})= \Phi_m(\ost{\ti}).$
Define $m'$ analogously for the set $\ostp{\ti}$.
The statement of Lemma~\ref{lem:inter1} now changes as follows:
 \begin{eqnarray}
 \notag
 \Delta(\ostp{\ti}) & \leq & \sum_{p \in \ost{\ti}} || c(p) - m'|| \ \leq \  \sum_{p \in \ost{\ti}} || c(p) - m^\star|| \ \leq \ \sum_{p \in \ost{\ti}} \left( || c(p) - p|| + ||p-m^\star|| \right) \\
 \label{eq:inter1}
& = & \Phi_{C_i}(\ost{\ti}) + \Delta(\ost{\ti})
 \end{eqnarray}

\noindent
Proof of Lemma~\ref{lem:inter2} also changes as follows: let $m''$ be as in the statement of this lemma. Then,
 \begin{eqnarray*}
\Phi_{m''}(\ost{\ti}) & = & \sum_{p \in \ost{\ti}} ||p-m''|| \ \leq \  \sum_{p \in \ost{\ti}} \left( ||p-c(p)|| + ||c(p)-m''|| \right) \\
& = & \Phi_{C_i}(\ost{\ti}) + \Phi_{m''}(\ostp{\ti}) \ \leq \ \Phi_{C_i}(\ost{\ti}) + \left(1 + \frac{\eps}{8} \right) \cdot \Delta(\ostp{\ti}) \\
&  \stackrel{\mbox{\small{(\ref{eq:inter1})}}}{\leq} & 2 \cdot \Phi_{C_i}(\ost{\ti}) + \left(1 + \frac{\eps}{8} \right) \cdot \Delta(\ost{\ti})
\ \stackrel{\mbox{\small{Lemma~\ref{lem:inter0}}}}{\leq} \ \frac{\eps}{3k} \cdot \opt_k(\bOst) + \left(1 + \frac{\eps}{8} \right) \cdot \Delta(\ost{\ti})
 \end{eqnarray*}

\noindent
Rest of the arguments remain unchanged (we use Lemma~\ref{lem:kss} instead of Lemma~\ref{lemma:inaba}).
Now we consider Case~II.  We redefine the parameter $R$ as
$$ R = \frac{\eps}{9} \cdot \frac{\Phi_{C_i}(\ost{\ti})}{|\ost{\ti}|}.$$
 Define sets $\ostp{\ti}, \ostpn{\ti}, \ostf{\ti}$ as before.
 Let $m^\star$ be the point for which $\Delta(\ost{\ti})= \Phi_{m^\star}(\ost{\ti})$, and $m'$ be the analogous point for $\ostp{\ti}$.
 Proof of Lemma~\ref{lem:IIinter0} can be easily modified to yield the following~(instead of Fact~\ref{lem:folklore}, we just need
to use triangle inequality)~:

\begin{eqnarray}
\label{eq:II0-km}
 \Delta(\ost{\ti}) = \Phi_{m^\star}(\ost{\ti}) \geq \frac{4 n}{\eps} \cdot R
\end{eqnarray}

\noindent
We have the following version of Lemma~\ref{lem:IIinter1}:
\begin{eqnarray}
\notag
\Delta(\ostp{\ti}) & \leq & \Phi_{m^\star}(\ostp{\ti}) \ = \ \sum_{p \in \ostn{\ti}} ||c(p)-m^\star||  + \sum_{p \in \ostf{\ti}} ||p-m^\star|| \\
\notag
& \leq & \sum_{p \in \ostn{\ti}} \left( ||p-m^\star||+ ||c(p)-p|| \right)  + \sum_{p \in \ostf{\ti}} ||p-m^\star|| \\
\label{eq:II-km}
& \leq & nR + \Delta(\ost{\ti}),
\end{eqnarray}
where $n$ denotes $|\ost{\ti}|$.
Finally, let $m''$ be as in the statement of Lemma~\ref{lem:II-final}. Then,
\begin{eqnarray}
\notag
\Phi_{m''}(\ost{\ti}) &=& \sum_{p \in \ostn{\ti}} ||p-m''||  + \sum_{p \in \ostf{\ti}} ||p-m''|| \\
\notag & \leq & \sum_{p \in \ostn{\ti}} \left( ||c(p)-m''|| + ||c(p)-p|| \right) + \sum_{p \in \ostf{\ti}} ||p-m''|| \\
\notag & \leq & nR + \Phi_{m''}(\ostp{\ti}) \ \leq \ nR + \left(1 + \frac{\eps}{8} \right) \cdot \Delta(\ostp{\ti}) \\
& \stackrel{(~\ref{eq:II-km})}{\leq} & 3nR + \left(1 + \frac{\eps}{8} \right) \cdot \Delta(\ost{\ti})
\  \stackrel{(~\ref{eq:II0-km})}{\leq} \ (1+\eps) \cdot \Delta(\ost{\ti}).
\end{eqnarray}
Rest of the arguments go through without any changes.

% Conclusion
\section{Conclusion}
\label{sec:concl}
We formulated the list $k$-means problem and gave nearly tight upper and lower bounds on the size of the list of candidate centers. We also obtained
an algorithm for the  constrained $k$-means problem getting a significant improvement over the previous results of Ding and Xu~\cite{dx15}.
%%We also argue that our algorithm is close to optimal with respect to the size of the list of solutions that such algorithms are expected to produce.
Furthermore, we show how our techniques generalize for the corresponding $k$-median problems.
We would also like to point out that our techniques generalize for settings that involve non-Euclidean distance measures.
After going through the analysis of our algorithm, it is not difficult to show that the only properties that are used in the analysis  are:
\begin{itemize}
\item[(i)] Symmetry of the distance measure (used implicitly)
\item[(ii)] (Approximate) Triangle Inequality: Fact~\ref{lem:triangle}
\item[(iii)] Centroid property: Fact~\ref{lem:folklore}
\item[(iv)] Sampling property: Lemma~\ref{lemma:inaba}
\end{itemize}
The analysis holds even for some approximate versions of the above properties. 
For instance, for the $k$-median problem we were able to use Lemma~\ref{lem:kss} instead of Lemma~\ref{lemma:inaba} (i.e., the sampling property).
Also, we were able to work without the centroid property since for the $k$-median problem the distances follow the exact triangle inequality instead of the approximate version (i.e., Fact~\ref{lem:triangle}).
We note that there are a number of clustering problems in machine learning that are modeled as $k$-median problem over distance measures that follow the above properties in some approximate sense.
Mahalanobis distance and $\mu$-similar Bregman divergence are two examples of such distance measures.
Our results can be very easily extended for the $k$-median problem over such distance measures \footnote{Please see \cite{jks} for a discussion on such distance measures. This work shows how to extend such $D^2$-sampling based analysis to settings involving such distance measures.}.

}

\bibliographystyle{alpha}
\bibliography{paper}

\end{document}